\documentclass[10pt]{llncs}

\usepackage{amsmath}
\usepackage{verbatim}

\begin{document}

\title{Fast Set Intersection and Two-Patterns Matching}
\author{Hagai Cohen \and Ely Porat\thanks{This work was supported by BSF and ISF}}
\institute{Department of Computer Science, Bar-Ilan University, 52900 Ramat-Gan, Israel
\email{\{cohenh5,porately\}@cs.biu.ac.il}}

\maketitle

\begin{abstract}
In this paper we present a new problem, the \emph{fast set intersection} problem,
which is to preprocess a collection of sets in order
to efficiently report the intersection of any two sets in the collection.
In addition we suggest new solutions for
the \emph{two-dimensional substring indexing} problem
and the \emph{document listing} problem for two patterns
by reduction to the \emph{fast set intersection} problem.
\end{abstract}

\section{Introduction and Related Work}
The intersection of large sets is a common problem in the context of retrieval algorithms,
search engines, evaluation of relational queries and more.
Relational databases use indices to decrease query time, but when a query involves two different indices,
each one returning a different set of results, we have to intersect these two sets to get the final answer.
The running time of this task depends on the size of each set,
which can be large and make the query evaluation take longer
even if the number of results is small.
In information retrieval there is a great use of inverted index as a major indexing structure
for mapping a word to the set of documents that contain that word.
Given a word, it is easy to get from the inverted index the set of all the documents that contain that word.
Nevertheless, if we would like to search for two words to get all documents that contain both,
the inverted index doesn't help us that much.
We have to calculate the occurrences set for each word and intersect these two sets.
The problem of intersecting sets finds its motivation also in web search engines
where the dataset is very large.

Various algorithms to improve the problem of intersecting sets have been introduced in the literature.
Demaine et al. \cite{Demaine00adaptiveset} proposed a method
for computing the intersection of $k$ sorted sets using an adaptive algorithm.
Baeza-Yates \cite{Baeza04} proposed an algorithm to improve the multiple searching problem
which is related directly to computing the intersection of two sets.
Barbay et al. \cite{Barbay06fasteradaptive} showed that using
interpolation search improves the performance of adaptive intersection algorithms.
They introduced an intersection algorithm for two sorted sequences that is fast on average.
In addition Bille et al. \cite{Philip07} presented a solution for computing expressions on given sets
involving unions and intersections.
A special case of their result is the intersection of $m$ sets containing $N$ elements in total,
which they solve in expected time $O(N (\log\omega)^{2}/\omega + m \cdot output)$
for word size $\omega$
where $output$ is the number of elements in the intersection.

In this paper we present a new problem, the \emph{fast set intersection} problem.
This problem is to preprocess a databases of size $N$
consisting of a collection of $m$ sets
to answer queries in which we are given two set indices $i,j \leq m$,
and wish to find their intersection.
This problem has lots of applications where there is a need to intersect two sets
in a lot of different fields
like Information Retrieval, Web Searching, Document Indexing, Databases etc.
An optimal solution for this problem will bring better solutions to various applications.

We solve this problem using minimal space and still decrease the query time
by using a preprocessing part.
Our solution is the first non-trivial algorithm for this problem.
We give a solution that requires linear space
with worst case query time bounded by $O(\sqrt{N output} + output)$
where $output$ is the intersection size.

In addition, we present a solution
for the \emph{two-dimensional substring indexing} problem,
introduced by Muthukrishnan et al. \cite{Muthu01}.
In this problem we preprocess a database $D$ of size $N$.
So when given a string pair $(\sigma_{1}, \sigma_{2})$,
we wish to return all the database string pairs $\alpha_{i} \in D$
such that $\sigma_{1}$ is a substring of $\alpha_{i,1}$
and $\sigma_{2}$ is a substring of $\alpha_{i,2}$.
Muthukrishnan et al. suggested a tunable solution for this problem which
uses $O(N^{2-y})$ space for a positive fraction $y$
and query time of $O(N^y+ output)$ where $output$ is the number of such string pairs.
We present a solution for this problem, based on solving the \emph{fast set intersection} problem,
that uses $O(N\log N)$ space
with $O((\sqrt{N\log N output} + output)\log^2 N)$ query time.

In the \emph{document listing} problem which was presented by Muthukrishnan \cite{Muthu02},
we are given a collection of size $N$ of text documents
which may be preprocessed
so when given a pattern $p$ we want to return the set of all the documents that contain that pattern.
Muthukrishnan suggested an optimal solution for this problem
which requires $O(N)$ space with $O(|p| + output)$ query time
where $output$ is the number of documents that contain the pattern.
However, there is no optimal solution
when given a query consists of two patterns $p, q$ to return the set of all the documents that contain them both.
The only known solution for this problem is of Muthukrishnan \cite{Muthu02}
which suggested a solution that uses $O(N\sqrt{N})$ space
which supports queries in time $O(|p| + |q| + \sqrt{N} + output)$.
We present a solution for the \emph{document listing} problem when the query consists of two patterns.
Our solution uses $O(N\log N)$ space
with $O(|p| + |q| + (\sqrt{N\log N output} + output)\log^2 N)$ query time.

The paper is structured as follows:
In Sect.~\ref{sec:FSI Problem} we describe the \emph{fast set intersection} problem.
In Sect.~\ref{sec:FSI Solution} we describe our solution for this problem.
In Sect.~\ref{sec:Intersection-Empty Query and Intersection-Size Query}
we present similar problems with their solutions.
In Sect.~\ref{sec:Two-Dimensional Substring Indexing Solution}
we present our solution for the \emph{two-dimensional substring indexing} problem
and the \emph{document listing} problem for two patterns.
In Sect.~\ref{sec:Conclusions} we present some concluding remarks.

\section{Fast Set Intersection Problem}\label{sec:FSI Problem}
We formally define the fast set intersection (FSI) problem.

\begin{definition}
Let $D$ be a database of size $N$ consisting of a collection of $m$ sets.
Each set has elements drawn from $1\ldots c$.
We want to preprocess $D$ so that given a query of two indices $i,j \leq m$,
we will be able to calculate the intersection between sets $i, j$ efficiently.
\end{definition}

A naive solution for this problem is to store the sets sorted.
Given a query of two sets $i, j$, go over the smaller set and check for each element if it exists in the second set.
This costs $O(min(|i|,|j|)\log{max(|i|,|j|}))$.
This solution can be further improved using hash tables.
A static hash table \cite{Fredman84} can store $n$ elements
using $O(n)$ space and build time, with $O(1)$ query time.
For each set we can build a hash table to check in $O(1)$ time if an element is in the set or not.
This way the query time is reduced to $O(min(|i|,|j|))$ using linear space.
The disadvantage of using this solution is that on the worst case
we go over a lot of elements even if the intersection is small.
A better query time can be gained by using more space
for saving the intersection between every two sets.
Using $O(m^{2}c)$ space we get an optimal query time of $O(output)$
where $output$ is the size of the intersection.
Nevertheless, this solution uses extremely more space.
In the next section we present our solution for the \emph{fast set intersection} problem
which bounds the query time on the worst case.

\section{Fast Set Intersection Solution}\label{sec:FSI Solution}
Here we present our algorithm for solving the FSI problem.
We call \emph{result set} to the output of the algorithm, i.e., the intersection of the two sets.
By $output$ we denote the size of the result set.

\subsection{Preprocessing}
For each set in $D$ we store a hash table to know in $O(1)$ time if an element is in that set or not.
In addition, we store the inverse structure,
i.e., for each element we store a hash table to know in $O(1)$ time if it belongs to a given set or not.

Our main data structure consists of an unbalanced binary tree.
Starting from the root node at level $0$,
each node in that tree handles number of subsets of the original sets from $D$.
The cost of a node in that tree is the sum of the sizes of all the subsets it handles.
The root node handles all the $m$ sets in $D$, therefore, it costs $N$.

\begin{definition}
Let $d$ be a node which costs $n$.
A \emph{large set} in $d$ is a set which has more than $\sqrt{n}$ elements.
\end{definition}

\begin{lemma}\label{lem:LargeSetsNumber}
By definition, a node $d$ which costs $n$,
can handle at most $\sqrt{n}$ large sets.
\end{lemma}

A \emph{set intersection matrix} is a matrix that stores for each set
if it has an intersection with any other set.
For $\grave{m}$ sets this matrix costs $O(\grave{m}^{2})$ bits space
with $O(1)$ query time for answering if set $i$ and set $j$ have a non-empty intersection.

For each node we construct a set intersection matrix for the large sets in that node.
By lemma~\ref{lem:LargeSetsNumber}, saving the set intersection matrix
only for the large sets in a node that costs $n$ space will cost only another $n$ space.

Now we describe how we divide sets between the children of a node.
Only large sets in a node will be propagated down to its two children,
we call them the \emph{propagated group}.
Let $d$ be a node which costs $n$ and let $G$ be its propagated group.
Then, $G$ costs at most $n$ as well.
Let $E$ be the set of all elements in the sets of $G$.
We partition $E$ into two disjoint sets $E_{1},E_{2}$.
For a given set $S \in G$ we partition it between the two children as following:
The left child will handle $S \cap E_{1}$
and the right child will handle $S \cap E_{2}$.
We want each child of $d$ to cost at most $\frac{n}{2}$.
Nevertheless, finding such a partition of $E$ is a hard problem,
if even possible at all.
To overcome this difficulty we shall add elements to $E_{1}$
until adding another element will make the left child cost more than $\frac{n}{2}$.
The next element, which we denote by $e$, will be remarked in $d$ for checking, during query time,
whether it lies in the intersection.
We now take $E_{2} = E - E_{1} - \{e\}$ , i.e., the remaining elements.
This way each child costs at most $\frac{n}{2}$.

A leaf in this binary tree is a node which is in constant size.
Because each node in the tree costs half the space of its parent then this tree has $\log N$ levels.

\begin{theorem}
The space needed for this data structure is $O(N)$ space.
\end{theorem}
\begin{proof}
The hash tables for all the sets cost $O(N)$ space.
As well the inverse hash tables for all the elements cost $O(N)$ space.

The binary tree structure space cost is as follows:
The root costs $O(N)$ bits for saving the set intersection matrix.
In each level we store only another $O(N)$ bits because every two children don't cost more than their parent.
Hence, the total cost of this tree structure is $O(N\log N)$ bits
which is $O(N)$ space in term of words.
\qed
\end{proof}

\subsection{Query Answering}
Given sets $i, j$ (without loss of generality we assume $|i| \leq |j|$),
we start traversing the tree from the root node.
If $i$ is not a large set in the root we check each element from it in the hash table of $j$.
As there can be at most $\sqrt{N}$ elements in $i$ because it is not a large set,
this will cost $O(\sqrt{N})$.
If both $i, j$ are large sets we do as follows:
We check in the set intersection matrix of the root wether there is a non-empty intersection between $i$ and $j$.
If there is not
there is nothing to add to the result set so we stop traversing down.
If there is an intersection
we check the hash table of the element which is remarked in that node if it belongs to $i$ and $j$
and add that element to the intersection if it belongs to both.
Next we go down to the children of the root and continue the traversing recursively.

Elements are added to the result set when we get to a node
which in that node $i$ is not a large set.
In this case, we stop traversing down the tree from that node.
Instead we step over all the elements of $i$ in that node checking for each one of them
if it belongs to $j$.
We call such a node a \emph{stopper node}.

\begin{theorem}
The query time is bounded by $O(\sqrt{N output}+ output)$.
\end{theorem}
\begin{proof}
The query computation consists of two parts.
The tree traversal part and the time we spend on stopper nodes.

There are $output$ elements in the result set, therefore, there can be at most $O(output)$ stopper nodes.
Because the tree height is $\log N$,
for each stopper node we visit at most $\log N$ nodes for the tree traversal until we get to it.
Therefore, the tree traversal part adds at most $O(output\log N)$ to the query time.
But this is more than what we actually pay for the tree traversal
because some stopper nodes share their path from the root.
This can be bounded better.
Because the tree is a binary tree if we fully traverse the tree till $\log output$ height
it will cost $O(output)$ time.
Now, from this height if we continue traverse the tree
we visit for each stopper node at most $\log N - \log output$ nodes because we are already at $\log output$ height.
Thus, the tree traversal part is bounded by $O(output + output (\log N - \log output))$.
By log rules this equals to $O(output + output\log \frac{N}{output})$.

Now, we calculate how much time we spent on all the stopper nodes.
A stopper node is a node which during the tree traversal we have to go over all elements of a non-large set in that node.
The size of a non-large set in a stopper at level $l$ is $\sqrt{\frac{N}{2^l}}$.
Consider there are $x$ stopper nodes.
We denote by $l_i$ the level for stopper node $i$.
For all stopper nodes we pay at most:
\begin{align*}
\sum_{i=1}^x \sqrt{\frac{N}{2^{l_i}}} = \sqrt{N} \sum_{i=1}^x 2^{-\frac{1}{2}l_i}
= \sqrt{N} \sum_{i=1}^x 1 \cdot 2^{-\frac{1}{2}l_i}
\end{align*}
The Cauchy-Schwarz inequality is that $(\sum_{i=1}^n x_i y_i)^2 \leq (\sum_{i=1}^n x_i^2)(\sum_{i=1}^n y_i^2)$.
We use it in our case to get:
\begin{align*}
&\leq \sqrt{N} \sqrt{\sum_{i=1}^x 1^2} \sqrt{\sum_{i=1}^x (2^{-\frac{1}{2}l_i})^2} \\
&= \sqrt{N} \sqrt{x} \sqrt{\sum_{i=1}^x 2^{-l_i}}
\end{align*}
Kraft inequality from Information Theory states that for any binary tree:
\begin{align*}
\sum_{l \in leaves} 2^{-depth(l)} \leq 1
\end{align*}
Because we never visit a subtree rooted by a stopper node,
then in our case each stopper node can be viewed as a leaf in the binary tree.
Therefore, we can transform Kraft inequality for all the stopper nodes
instead of all tree leaves to get that
$\sum_{i=1}^x 2^{-l_i} \leq 1$.
Using this inequality gives us that:
\begin{align*}
\leq \sqrt{N} \sqrt{x} = \sqrt{Nx} \leq \sqrt{N output} = output\sqrt{\frac{N}{output}}
\end{align*}
Thus, we pay $O(output\sqrt{\frac{N}{output}})$, for the time we spend in the stopper nodes.

Therefore, the tree traversal part and the time we spend on all stopper nodes is
$O(output + output\log \frac{N}{output} + output\sqrt{\frac{N}{output}})$.
Hence, the final query time is bounded by $O(\sqrt{N output}+ output)$.
\qed
\end{proof}

%

\begin{corollary}
The fast set intersection problem can be solved in linear space with worst case query time of $O(\sqrt{N output}+ output)$.
\end{corollary}

\section{Intersection-Empty Query and Intersection-Size Query}\label{sec:Intersection-Empty Query and Intersection-Size Query}
In the FSI problem given a query we want to return
the result set, i.e., the intersection between two sets.
What if we only want to know if there is any intersection between two sets?
We call that the \emph{intersection-empty query} problem.
Moreover, sometimes we would like only to know the size of the intersection
without calculating the actual result set.
We define these problems as follows:
\begin{definition}
Let $D$ be a database of size $N$ consisting of a collection of $m$ sets.
Each set has elements drawn from $1\ldots c$.
The \emph{intersection-empty query} problem is to preprocess $D$ so that given a query of two indices $i,j \leq m$,
we want to calculate if sets $i, j$ have any intersection.
In the \emph{intersection-size query} problem when given a query we want to calculate the size of the result set.
\end{definition}


A naive solution for the intersection-empty query problem
is to build a matrix saving if there is any intersection between every two sets.
This solution uses $O(m^2)$ bits space with query time of $O(1)$.
For the intersection-size query problem we store the intersection size for every two sets
by using slightly more space, $O(m^2)$ space, with query time of $O(1)$.

We can use part of our FSI solution method to solve the intersection-empty query problem
using $O(N)$ space with $O(\sqrt{N})$ query time.
Instead of the whole tree structure we store only the root node with its set intersection matrix using $O(N)$ space.
Given sets $i, j$ (without loss of generality let's assume $|i| \leq |j|$),
if $i$ is not large set in the root we check each element from it in the hash table of $j$.
Because $i$ is not large set, this will cost at most $O(\sqrt{N})$ time.
If $i$ is a large set then we check in the set intersection matrix of the root
to see if there is any intersection in $O(1)$ time.
Hence, we can solve the intersection-empty query problem in $O(\sqrt{N})$ time using $O(N)$ space.

With the same method we can solve the intersection-size query problem
by saving the size of the intersection instead of saving if there is any intersection in the set intersection matrix.
This way we can solve the intersection-size query problem in $O(\sqrt{N})$ time using $O(N)$ space.

%
%
%

\section{Two-Dimensional Substring Indexing Solution}\label{sec:Two-Dimensional Substring Indexing Solution}
In this section,
we show how to solve the \emph{two-dimensional substring indexing} problem
and the \emph{document listing} problem for two patterns
using our FSI solution.
The \emph{two-dimensional substring indexing} problem was showed by Muthukrishnan et al. \cite{Muthu01}.
It is defined as follows:

\begin{definition}
Let $D$ be a database consisting of a collection of string pairs
$\alpha_{i} = (\alpha_{i,1},\alpha_{i,2}), 1\leq i \leq c$, which may be preprocessed.
Given a query string pair $(\sigma_{1}, \sigma_{2})$,
the 2-d substring indexing problem is to identify all string pairs $\alpha_{i} \in D$,
such that $\sigma_{i}$ is a substring of $\alpha_{i,1}$ and $\sigma_{2}$ is a substring of $\alpha_{i,2}$.
\end{definition}

Muthukrishnan et al. \cite{Muthu01} reduced the \emph{two-dimensional substring indexing} problem
to the \emph{common colors query} problem which is defined as follows:

\begin{definition}
We are given an array $A[1 \ldots N]$ of colors drawn from $1 \ldots C$.
We want to preprocess this array so that the following query can be answered efficiently:
Given two non-overlapping intervals $I_{1}, I_{2}$ in $[1,N]$,
list the distinct colors that occur in both intervals $I_{1}$ and $I_{2}$.
\end{definition}

The common colors query (CCQ) problem is another intersection problem
where we have to intersect two intervals on the same array.
We now show how to solve the CCQ problem by solving the FSI problem.
By that we solve the \emph{two-dimensional substring indexing} problem as well.

Given array $A$ of size $N$,
we build a data structure consisting of $\log N$ levels over this array.
In the top level we partition $A$ into two sets of size at most $\frac{N}{2}$,
the first set containing colors, i.e., elements, of $A$ in range $A[1 \ldots \frac{N}{2}]$
and the second set containing colors in range $A[\frac{N}{2}+1 \ldots N]$.
As well, each level $i$ is partitioned into $2^i$ sets,
each respectively, containing a successive set of $\frac{N}{i}$ colors from $A$.
The bottom level, in similar fashion,
is therefore partitioned into $N$ sets each containing one different color from array $A$.
The size of all the sets in each level is $O(N)$.
Therefore, the size needed for all the sets in all levels is $O(N\log N)$.

\begin{lemma}\label{lem:IntervalCovering}
An interval $I$ on $A$ can be covered by at most $2\log N$ sets.
\end{lemma}

\begin{proof}
Assume, by contradiction, that there exists an interval for which at least $m > 2\log N$ sets are needed.
This implies that there is some level that at least $3$ (consecutive) sets are selected.
However, for every $2$ consecutive sets there have to be a set in the upper level
that contains them both, so we can take it instead,
and cover the same interval with only $m-1$ sets,
in contradiction to the assumption that at least $m$ sets are required for the cover.
\qed
\end{proof}

\begin{theorem}\label{thm:CCQByFSI}
The CCQ problem can be solved
using $O(N\log N)$ space with $O((\sqrt{N\log N output} + output)\log^2 N)$ query time
where $output$ is the number of distinct colors that occur in both $I_{1}$ and $I_{2}$.
\end{theorem}
\begin{proof}
Given two intervals $I_{1},I_{2}$ we want to calculate their intersection,
By lemma~\ref{lem:IntervalCovering}, $I_{1}, I_{2}$ are each covered by a group of $2\log n$ sets at the most.
To get the intersection of $I_{1},I_{2}$ we will take each set from the first group
and intersect it with each set from the second group using our FSI solution.
Hence, we have to solve the FSI problem $O(\log^2 N)$ times.
Our FSI solution takes $O(\sqrt{N output} + output)$ time and $O(N)$ space for dataset which costs $O(N)$ space.
Here the dataset costs $O(N\log N)$ space,
therefore, we can solve the common colors query problem
in $O((\sqrt{N\log N output} + output)\log^2 N)$ time using $O(N\log N)$ space.
\qed
\end{proof}

As showed in \cite{Muthu01} to solve the two-dimensional substring problem we can solve a CCQ problem.
As a result, the two-dimensional substring problem can be solved
in $O((\sqrt{N\log N output} + output)\log^2 N)$ time using $O(N\log N)$ space.

\subsection{Document Listing Solution For Two Patterns}
The document listing problem was presented by Muthukrishnan \cite{Muthu02}.
In this problem we are given a collection $D$ of text documents $d_{1}, \ldots, d_{c}$,
with $\sum_{i}|d_{i}| = N$, which may be preprocessed,
so when given a query comprising of a pattern $p$
our goal is to return the set of all documents that contain one or more copies of $p$.
Muthukrishnan presented an optimal solution for this problem by building a suffix tree for $D$,
searching the suffix tree for $p$ and getting an interval $I$ on an array with all the occurrences of $p$ in $D$.
Then they solve the colored range query problem on $I$ to get each document only once.
This solution requires $O(N)$ space with optimal query time of $O(|p| + output)$
where $output$ is the number of documents that contain $p$.

We are interested in solving this problem for a two patterns query.
Given two patterns $p, q$,
our goal is to return the set of all documents that contain both $p$ and $q$.
In \cite{Muthu02} there is a solution that uses $O(N\sqrt{N})$ space with $O(|p| + |q| + \sqrt{N} + output)$ query time.
Their solution is based on searching a suffix tree of all the documents
for the two patterns $p, q$ in $O(|p| + |q|)$ time.
From this they get two intervals:
$I_{1}$ with $p$ occurrences and $I_{2}$ with $q$ occurrences..
On these intervals they solve a CCQ problem to get the intersection
between $I_{1}$ and $I_{2}$ for all the documents that contain both $p$ and $q$.

We suggest a new solution based on solving the FSI problem.
We use the same method as Muthukrishnan \cite{Muthu02}
until we get the two intervals:
$I_{1}$ with $p$ occurrences and $I_{2}$ with $q$ occurrences.
Now, we have to solve a CCQ problem which can be solved as shown above in theorem~\ref{thm:CCQByFSI}.
Therefore, the document listing problem for two patterns can be solved
in $O(|p| + |q| + (\sqrt{N\log N output} + output)\log^2 N)$ time using $O(N\log N)$ space
where $output$ is the number of documents that contain both $p$ and $q$.

\section{Conclusions}\label{sec:Conclusions}
In this paper we developed a method to improve algorithms which intersects sets as a common task.
We solved the fast set intersection problem
using $O(N)$ space with query time bounded by $O(\sqrt{N output} + output)$.
We showed how to improve some other problems,
the two-dimensional substring indexing problem and the document listing problem for two patterns,
using the fast set intersection problem.

There is still a lot of research to be done in regards to the fast set intersection problem.
It is open if the query time can be bounded better.
Moreover, we showed only two applications for the fast set intersection problem.
We are sure that the fast set intersection problem can be useful in other fields as well.

\bibliographystyle{splncs}
\bibliography{two_queryBib}

\end{document}